\documentclass[aps,pra,twocolumn,superscriptaddress,longbibliography]{revtex4-1}
\usepackage{amsthm}
\usepackage[utf8]{inputenc}
\usepackage[polish,english]{babel}
\usepackage[T1]{fontenc}
\usepackage{indentfirst}
\usepackage{graphicx}
\usepackage{amsmath,bbm}
\usepackage{url,hyperref}
\usepackage{color}
\usepackage{todonotes}

\newcommand{\bra}[1]{\langle #1 |}
\newcommand{\ket}[1]{| #1 \rangle}
\newcommand{\braket}[2]{\langle #1 | #2 \rangle}
\newcommand{\Tr}{\hbox{Tr}}

\def\C{{\mathbbm C}}

\def\N{{\mathbbm N}}
\def\P{{\mathbbm P}}
\def\({\left(}
\def\){\right)}
\newcommand{\vi}{\vec{n}}
\newtheorem{theorem}{Theorem}

\begin{document}

\title{Which multiphoton states are related via linear optics?}

\author{Piotr Migda{\l}}
\email{pmigdal@gmail.com}
\homepage{http://migdal.wikidot.com}
\affiliation{ICFO--Institut de Ci\`{e}ncies Fot\`{o}niques, 08860
  Castelldefels (Barcelona), Spain}

\author{Javier Rodr\'{\i}guez-Laguna}
\affiliation{Mathematics Department, Universidad Carlos III de Madrid, Spain}
\affiliation{ICFO--Institut de Ci\`{e}ncies Fot\`{o}niques, 08860
  Castelldefels (Barcelona), Spain}

\author{Micha{\l} Oszmaniec} 
\affiliation{Center for Theoretical Physics, Polish Academy of
  Sciences, Warszawa, Poland}

\author{Maciej Lewenstein}
\affiliation{ICFO--Institut de Ci\`{e}ncies Fot\`{o}niques, 08860
  Castelldefels (Barcelona), Spain}
\affiliation{ICREA--Instituci\'o Catalana de Recerca i Estudis
  Avan\c{c}ats, Barcelona, Spain}

\date{\today}
\begin{abstract}
We investigate which pure states of $n$ photons in $d$ modes can be transformed into each other via linear optics, without post-selection. In other words, we study the local unitary (LU) equivalence classes of symmetric many-qudit states. Writing our state as $f^\dagger\ket{\Omega}$, with $f^\dagger$ a homogeneous polynomial in the mode creation operators, we propose two sets of LU-invariants: (a) {\em spectral invariants}, which are the eigenvalues of the operator $ff^\dagger$, and (b) {\em moments}, each given by the norm of the symmetric component of a tensor power of the initial state, which can be computed as vacuum expectation values of $f^k(f^\dagger)^k$.
We provide scheme for experimental measurement of the later,
as related to the post-selection probability of creating state $f^{\dagger k}\ket{\Omega}$ from $k$ copies of $f^{\dagger}\ket{\Omega}$.
\end{abstract} 

\maketitle

\section{Introduction}

The characterization of entanglement is one of the most relevant problems of quantum theory \cite{Horodecki2009}. Bipartite entanglement for pure states is well understood, due to the invariance under local operations of the {\em entanglement spectrum}. Multipartite entanglement, on the other hand, is a challenging problem. An equivalent formulation of the problem is the question about which pure states can be mutually transformed via local unitary transformations, or {\em LU-equivalence}.

In the field of {\em quantum optics} this question can be recast in this way: can a given multi-photon state, $\ket{\psi_1}$, be transformed into another one, $\ket{\psi_2}$, using only {\em linear optics}? By linear (or {\em passive}) optics we mean the use of beam-splitters and wave plates, which is known to be equivalent to the action of arbitrary unitary operations on each mode \cite{Lee2002}. This question bears special relevance both in theory and practical applications. On the theoretical side, linear optics with post-selection has been proved to be able to efficiently realize a universal quantum computer \cite{Knill2001, Kok2007}. But even without post-selection, linear optics transformations of multi-photon states constitute an intermediate stage between classical and full-fledged quantum computation \cite{Aaronson2010}. On the practical side, our ability to generate decoherence-free states \cite{Bourennane2004, Wu2011} relies on our ability to transform multi-photon states.
Operation by linear optics can be viewed as multi-particle interference, as opposed to multi-particle interaction. But beyond a generic interference phenomenon, it bears specific effects which are specific to bosons \cite{Tichy2012,Tichy2013}.

This paper considers the equivalence under linear optics transformations of pure states of $n$ photons in $d$ modes, disregarding the possibility of post-selection. Thus, the problem is framed as one of LU-equivalence of bosonic states \cite{Migdal2013aaa}.

As an illustration, let us consider a state of two photons occupying two different modes or channels, $\ket{\psi_1}=\ket{1,1}$. It is possible to transform this state into $\ket{\psi_2}=(\ket{2,0}-\ket{0,2})/\sqrt{2}$ using Hong-Ou-Mandel interference \cite{Hong1987} (i.e.: two-photon interference in a $50\%:50\%$ beam splitter), but it is not possible to place both photons in the same channel with $100\%$ efficiency, i.e.: $\ket{\psi_3}=\ket{2,0}$ is {\em not} achievable. Of course, it is always possible to perform post-selection, measuring the number of photons in the second channel and retaining only the states that contain none, but the efficiency will drop to $50\%$.

Solutions of the LU-equivalence problem for distinguishable particles have been found for systems of a few particles \cite{Sudbery2000, DeVicente2012}. General approaches \cite{Carteret2000, Acin2001, Verstraete2003, Kraus2010a, Kraus2010b} involve the search for a standard form, which is hard to calculate analytically. In our work we present two sets of LU invariants, i.e. functions on the Hilbert space, which are invariant under the action of linear optics.

In this paper we focus on specific methods for bosonic states that provide analytic invariants. These invariants are built upon $f^\dagger$, the homogeneous polynomial on the creation operators which transforms the vacuum into our state. We present two families of LU-invariants, i.e.: two sets of complex-valued functions on the Hilbert space which are invariant under linear optics:

\begin{itemize}
\item The spectrum of the operator $f f^\dagger$.
\item The moments: vacuum expectation values of the operators $f^k f^{\dagger k}$, for any natural $k$.
\end{itemize}

The considered invariants are both simple to calculate and, as we will show, sufficient to distinguish states in many practical situations, even some states which are generally difficult to handle.

The article is organized as follows.
In Section \ref{s:Problem} we formally introduce the problem of LU-equivalence of bosonic states. In that part we also discuss briefly two simple settings: two particles (in an arbitrary number of modes) and two modes (containing any number of particles). 
In Section~\ref{s:ff} we present the construction and relevance of spectral invariants related to the operator $f f^\dagger$.
We show that, despite being infinite dimensional, this operator can be easily diagonalized, as it separates into blocks of fixed numbers of particles (not related to the photon count) which are related to many-body correlators.
In Section~\ref{s:replicas} we discuss the second set of invariants: vacuum expectation values of $f^k f^{\dagger k}$.
It corresponds to the projection of the tensor power of $k$ copies of our state (in the particle basis) onto the completely symmetric Hilbert space.
In Section~\ref{s:mode-product} we propose an interferometric scheme that, in principle, allows for a direct measurement of this set of invariants.
Moreover, such scheme allows direct experimental creation of states given by the polynomial $f^{k}$ for an arbitrary $k$.
In Section~\ref{s:examples} we apply our methods in concrete examples. We show that, using our invariants, we can solve the LU-equivalence problem for two particles in two modes and for three particles in two modes. We also study which states from the four-particle singlet subspace can be reached using linear optics from another state in the same singlet subspace.
Moreover, we show that, at least in some cases, $k$-particle blocks of $f f^\dagger$ provide more invariants than $k$-particle reduced density matrices.
Finally, Section~\ref{s:conclusion} concludes and points out some further directions.
Some technical discussions are left for the appendices:
Appendix~\ref{app:majorana} solves the LU-equivalence for symmetric qubit states (i.e. $n=2$), using Majorana representation. 
In Appendix~\ref{app:schwinger} we introduce Schwinger-like representation for expressing arbitrary $k$-body correlations in terms of normally ordered creation and annihilation operators.


\section{Formulation of the problem} 
\label{s:Problem}

Let us consider the system of $n$-photons in $d$ modes. There are, at least, two possible descriptions of the Hilbert space ${\cal S}_n^d$ describing the system. In a {\em mode description}, i.e.: the second quantization picture (see for example \cite{FetterBook1971}), ${\cal S}_n^d$ is treated as a subspace of the full Fock space $\mathrm{Fock}\left(\C ^d \right)$. Let $\vec n\equiv(n_1,\cdots,n_d)$ be a multi-index denoting the photon count for each mode and let $|\vec n| = \sum_{k=1}^d n_k$. The basis states spanning ${\cal S}_n^d$ are specified by the photon count on each mode,
\begin{equation}
\ket{\vec n} =
\frac{(a^{n_1}_1)^\dagger\cdots(a^{n_d}_d)^\dagger}{\sqrt{(n_1!)\cdots(n_d!)}}
\ket{\Omega} 
\equiv 
{\tilde a}^\dagger_{\vec n} \ket{\Omega},\ |\vec n| =n \ .
\label{eq:multia}
\end{equation}
In the above expression $\ket{\Omega}$ is the Fock vacuum, $a_1,\ldots,a_d$ are annihilation operators, and ${\tilde a}^\dagger_{\vec n}$ is a normalized monomial defined as above, creating $\ket{\vec n}$ from vacuum.

In a {\em particle description}, Hilbert space ${\cal S}_n^d$ is treated as the permutation-symmetric subspace of $(\C^d)^{\otimes n}$,  $\mathrm{Sym}^n \left(\C ^d\right)$. Let us fix the basis vectors of $\C ^d$: $\ket{1},\ket{2},\ldots,\ket{d}$. Basis states of $\left(\C ^d\right)^{\otimes n}$ with a simple tensorial form,
\begin{equation}
\ket{\phi}=\ket{i_1}\otimes \ket{i_2}\otimes \cdots \otimes \ket{i_n} \ ,\ i_k \in \left\{1,\ldots,d\right\}\ ,
\end{equation}
are not {\em permutation symmetric}. A basis for $\mathrm{Sym}^n\left(\C ^d\right)$ is obtained from product vectors in $(\C^d)^{\otimes n}$ by symmetrization over all factors in the tensor product. Let us define an asymmetric state from given mode counts $\vi=\{n_1,\cdots,n_d\}$:
\begin{equation}
\ket{\vi}_A \equiv \ket{1}_P^{\otimes n_1} \otimes \ket{2}_P^{\otimes n_2}
\otimes \cdots \otimes \ket{d}_P^{\otimes n_d}\ .
\label{eq:ordered-state}
\end{equation}
In the above expression we explicitly put the subscript $P$ to emphasize that we deal with tensor product of states in particle representation. The state $\ket{\vi}_A$ can be thought as a naive state in particle representation with the corresponding photon counts for each mode, but with a non-physical identification between particles and modes. The corresponding normalized symmetric state is given by:
\begin{equation}
\ket{\vi} = 
\frac{\sqrt{n_1! \ldots n_d!}}{\sqrt{n!}}
\sum_{\mathrm{perm}} \ket{1}_P^{\otimes n_1} \otimes \ket{2}_P^{\otimes n_2}
\otimes \cdots \otimes \ket{d}_P^{\otimes n_d}\ ,
\label{eq:fock_in_particle_representation}
\end{equation}
where the sum is over the different permutations of the factors appearing in the tensor product. Notice the required normalization factor. There exists another way of expressing the state $\ket{\vi}$ in particle basis
\begin{equation}
\ket{\vi} = N(\vi)\; \P^{(n)}_{sym} \ket{\vi}_A \ ,
\end{equation}
where $\P^{(n)}_{sym}$ is the projector onto the completely symmetric subspace of $(\C^d)^{\otimes n}$ and the normalization factor $N(n)$ is given by
\begin{equation}
N(n) = \sqrt{\frac{n!}{n_1!\cdots n_d!}}\ .
\end{equation}

During most of this work, we will work within the mode description, as it is more natural for dealing with boson states. However, in some parts of this paper we will use also the particle representation, and we will proceed between them both when it is convenient. States written in the particle representation will have a subscript $P$. States written in mode representation will have commas between modes.

An arbitrary pure state of the system can be written as: 
\begin{equation}
\ket{\psi} = \sum_{|\vec n|=n} \alpha_{\vec n}\; \ket{\vec n} \ ,
\label{def.f}
\end{equation}
where $\alpha_{\vec n}$ are complex amplitudes and $\ket{\vec n}$ are normalized states with fixed number of photons in each mode. To each state $\ket{\psi}$ we associate a unique homogeneous polynomial in the creation operators according to the recipe:
\begin{equation}
\ket{\psi}=\sum_{|\vec n|=n} \alpha_{\vec n}\; \ket{\vec n}= f^\dagger\ket{\Omega}  
\quad\to\quad
f^\dagger \equiv 
\sum_{|\vec n|=n} \alpha_{\vec n} \; {\tilde a}^\dagger_{\vec n} \ .
\label{eq:state-polynomial}
\end{equation}

In what follows we describe action of linear (or passive) optics on pure states described in different representations. Within the mode representation, action of linear optics is mathematically expressed in the application of {\em unitary operations} on the creation operators, i.e.:
\begin{align}
a_i'^\dagger = \sum_{j=1}^{d} U_{ij} a_j^\dagger \ ,
\label{u.transf}
\end{align}
where $U \in SU(d)$. Conversely, all $SU(d)$ operations among the modes can be achieved with a sequence of two-mode operations, such as beam-splitters and wave plates, in a way which resembles the action of Euler angles \cite{Reck1994}. Alternatively, in particle representation, transformation \eqref{u.transf}, is equivalent to the action of the same $U$ on each particle:
\begin{equation}
\ket{\psi'}_P = U^{\otimes n} \ket{\psi}_P \ .
\label{u.transf2}
\end{equation}
The equivalence between both representations corresponds to the equivalence between first and second quantization pictures for bosonic states \cite{FetterBook1971}.

We are now ready to state the problem of equivalence between two bosonic pure states under the action of linear optics. The problem is formulated as follows. Given two pure states, $\ket{\psi_1}=f_1^\dag \ket{\Omega}$ and $\ket{\psi_2}=f_2^\dag \ket{\Omega}\in {\cal S}_n^d$, whether there exists a unitary transformation on the modes $U\in SU(d)$ such that $f_1$ and $f_2$ are related by a rotation among the variables
\begin{equation}
f_2(\vec{a})^\dagger = f_1(U^\dagger \vec{a})^\dagger \ . \label{lu.equivalence}
\end{equation}

Alternatively, in the particle description, \eqref{lu.equivalence} is equivalent to
\begin{equation}
\ket{\psi_2}_P = U^{\otimes n} \ket{\psi_1}_P \ .
\label{lu.equivalence1}
\end{equation}

Interestingly, as it was recently showed, condition \eqref{lu.equivalence1} is equivalent to the local unitary equivalence (or LU-equivalence) of pure states, when restricted to permutation-symmetric states \cite{Migdal2013aaa}. That is, for $\ket{\psi_1}_P$ and $\ket{\psi_2}_P$ permutation-symmetric, it is equivalent to the existence of unitary operators $\{U_i\}_{\{1,\ldots,n\}}$ such that $\ket{\psi_2}_P = U_1 \otimes \ldots \otimes U_n \ket{\psi_1}_P$.

As it was stated in the introduction, our approach to the equivalence problem \eqref{lu.equivalence} is based on the construction of particular classes of invariants of the local unitary group representing linear optics. Let us consider the action of a group $G$ on some set $X$. For $x\in X$ and $g\in G$, let us denote the action of $g$ on $x$ by $g\cdot x$, which again belongs to $X$. A function $h: X\mapsto X$ is invariant under the action of $G$ if and only if
\begin{equation}
h(g\cdot x)=h(x)\ \text{for all $x\in X$ and all $g\in G$\ .}
\end{equation}
In our case we have $X=\mathcal{S}^{d}_n$, $G=SU(d)$ and the action of $G$ is given by \eqref{u.transf} or, equivalently, by \eqref{u.transf2}. A theorem by Hilbert states that, for a compact group $G$ acting in a unitary fashion on a finite dimensional vector space, there exists a finite number of independent invariants (which are polynomial in the coordinates of $\ket{\psi}$) that are able to distinguish whether two vectors belong to the same orbit of $G$ \cite{Weyl1997, Kraft1996, Grassl1998}.
A convenient way to write down the invariants involves using tensor diagrams \cite{Biamonte2013invariants} --- they make it explicit why certain polynomials are invariant and allow us to avoid multiple index contractions.
Thus, the LU-equivalence problem can be solved completely once the minimal set of independent polynomial invariants is known. This problem is in general unsolved. For a recent developments in theory of invariants in the context of entanglement theory see \cite{Vrana2011}. In our paper we do not attempt to study all invariants of the action of $SU\left(d\right)$ on $\mathcal{S}_{n}^d$. Instead, we focus on two families of invariants, analyzing their usefulness and physical relevance.

\subsection*{Simple examples}

Before considering the general problem, let us focus on two simple cases:
\begin{itemize}
\item only two particles ($n=2$) in an arbitrary number of modes,
  or, alternatively,
\item an arbitrary number of particles in just two modes ($d=2$)\ ,
  i.e.: permutation-symmetric states for qubits.
\end{itemize}

For two particles it suffices to perform a variant of the Schmidt decomposition, for symmetric states \cite{Li2001}, i.e.:
\begin{align}
\ket{\psi} = \sum_{i=1}^{d} \lambda_i \ket{\phi_i}_P \otimes \ket{\phi_i}_P \ ,
\end{align}
where $\lambda_i\geq 0$ and $\ket{\phi_i}$ are pairwise orthogonal states, the same for both particles. Thus, two pure states of two photons are related by linear optics if and only if they have the same sets of Schmidt values $\{\lambda_i\}$. In this case, $f$ is formally a quadratic polynomial in the number of modes, $d$. The Schmidt decomposition allows us to rewrite it as:
\begin{equation}
f = \sum_{i=1}^d \frac{\lambda_i}{\sqrt{2}} b_i^2
\label{factorize.n2}
\end{equation}
for a certain set of $\vec{b} = U \vec{a}$.

When there are just two modes ($d=2$), it is possible to use the {\em Majorana stellar representation} (see e.g. \cite{Aulbach2011a, Markham2010} for a short introduction) and write the state as:
\begin{equation}
\ket{\psi} = A \prod_{i=1}^{n} \left(\cos\({\theta_i\over 2}\) a_1^\dagger
+ e^{i \varphi_i} \sin\({\theta_i\over 2}\) a_2^\dagger \right) \ket{\Omega},
\label{majorana}
\end{equation}
where pairs $(\theta_i, \varphi_i)$ can be interpreted as coordinates of points on the Bloch sphere, and $A$ is a normalizing factor. Equation (\ref{majorana}) is equivalent to a factorization of the homogeneous polynomial defined in eq. (\ref{eq:state-polynomial}) in the following form:
\begin{equation}
f(a_1,a_2)= \tilde{A} a_2^n
\prod_{i=1}^{n}
\left( \tfrac{a_1}{a_2} - x_i \right)
\label{factorize.d2},
\end{equation}
where $\tilde{A}$ is the coefficient of $a_1^n$, and we have introduced variables $x_i = - e^{i \varphi_i} \tan(\theta_i/2)$.

Linear optics acts on this representation as a rotation of the Bloch sphere as a whole. Consequently, two states are related by linear optics if and only if their Majorana representations are related by rotation \cite{Mathonet2010}. A straightforward equivalence test based on this criterion is discussed in Appendix \ref{app:majorana}.


\section{Spectral method}\label{s:ff}

Let us consider the $d$-mode, $n$-particle bosonic state given in equation (\ref{def.f}), $\ket{\psi}=f^\dagger\ket{\Omega}$, where $f(a_1,\cdots,a_d)$ is a homogeneous polynomial of degree $n$ in the annihilation operators for the modes. Now, let us consider the operator $ff^\dagger$.

We will show that:
\begin{itemize}
\item its spectrum is invariant with respect to $SU(d)$ transformations (\ref{u.transf}),
\item it decomposes into an infinite number of blocks of finite size, but
\item the first $n$ blocks suffice to reconstruct the state.
\end{itemize}

\subsection{Invariance of the spectrum}
\begin{theorem}
The spectrum of $ff^\dagger$ is invariant with respect to arbitrary rotations between the modes, that is,
\begin{equation}
\text{Sp}\left[f(\vec{a}) f^\dagger(\vec{a})\right] =
\text{Sp}\left[f(U\vec{a}) f^\dagger(U\vec{a})\right]
\label{ff.invariance}
\end{equation}
for every $U\in SU(d)$. 
\end{theorem}
\begin{proof}
Each unitary operator acting on the modes $U = \exp(i H)$ (with Hermitian $H$) can be promoted to act on the full Fock-space via a second quantization extension:
\begin{align}
\tilde{U} = \exp\(i \sum_{i,j=1}^{d} H_{ij} a_i^\dagger a_j \),
\end{align}
where $\tilde{U}\cong U^{\otimes n}$ on our Hilbert space ${\cal S}_n^d$. This operator $\tilde{U}$ is unitary and acts on monomials in a natural way, i.e.: $\tilde{U}^\dagger a_j \tilde{U} = \sum_i U_{ji} a_i$, which can be checked with the Hadamard lemma. Consequently,
\begin{align}
f(U\vec{a}) f^\dagger(U\vec{a}) = \tilde{U}^\dagger f(\vec{a}) f^\dagger(\vec{a}) \tilde{U},
\end{align}
i.e.: the two operators are unitarily related and, thus, they have the same spectrum.
\end{proof}

\subsection{Block Decomposition}
\label{s:block-decomposition}

Since operator $f$ is a {\em homogeneous} polynomial of degree $n$ on the annihilation operators, each summand in operator $ff^\dagger$ contains $n$ creation and $n$ annihilation operators. Thus, $ff^\dagger$ preserves the number of photons $k$, and decomposes into blocks $ff^\dagger|_k$. Let $\vec k$ and $\vec k'$ be multi-indices with $|\vec k|=|\vec k'|=k$. Then, matrix elements of $ff^\dagger|_k$ can be shown to correspond to {\em correlators} of our state:
\begin{align}
&\bra{\vec k'}\; f f^\dagger\; \ket{\vec k} =
\bra{\Omega}\; \tilde a_{\vec k'}\; 
f f^\dagger\; 
\tilde a_{\vec k}^\dagger \; \ket{\Omega}  \nonumber\\ =& 
\bra{\Omega}\; f\; \tilde a_{\vec k'} \tilde a^\dagger_{\vec  k} \; 
f^\dagger\; \ket{\Omega} = \bra{\psi} \; \tilde a_{\vec k'} \tilde
a^\dagger_{\vec k}\; \ket{\psi}.
\label{eq:spec2correl}
\end{align}

For example, for two modes and particle numbers $k\in\{0, 1,2\}$, the blocks are given by:
\begin{align}
f f^\dagger|_{k=0} &=
\left[%
\begin{matrix}
\bra{\psi}1\ket{\psi}
\end{matrix}
\right]
\\
f f^\dagger|_{k=1} &=
\left[%
\begin{matrix}
\bra{\psi}a_1 a_1^\dagger\ket{\psi} & \bra{\psi}a_1 a_2^\dagger\ket{\psi} \\
\bra{\psi}a_2 a_1^\dagger\ket{\psi} & \bra{\psi}a_2 a_2^\dagger\ket{\psi}
\end{matrix}
\right]
\\
f f^\dagger|_{k=2} &=
\end{align}%
\begin{equation}
\left[%
\begin{matrix}
\bra{\psi} \frac{a_1^2 a_1^{\dagger 2}}{2} \ket{\psi}
& \bra{\psi} \frac{a_1^2 a_1^\dagger a_2^\dagger}{\sqrt{2}} \ket{\psi}
& \bra{\psi} \frac{a_1^2 a_2^{\dagger 2}}{2} \ket{\psi}
\\
\bra{\psi} \frac{a_1 a_2 a_1^{\dagger 2}}{\sqrt{2}} \ket{\psi}
& \bra{\psi} a_1 a_2 a_1^\dagger a_2^\dagger \ket{\psi}
& \bra{\psi} \frac{a_1 a_2 a_2^{\dagger 2}}{\sqrt{2}} \ket{\psi}
\\
\bra{\psi} \frac{a_2^2 a_1^{\dagger 2}}{2} \ket{\psi}
& \bra{\psi} \frac{a_2^2 a_1^\dagger a_2^\dagger}{\sqrt{2}} \ket{\psi}
& \bra{\psi} \frac{a_2^2 a_2^{\dagger 2}}{2} \ket{\psi}
\end{matrix}
\right]\nonumber.
\end{equation}

The matrix elements of $ff^\dagger|_k$ are $k$-particle correlators. For $k=0$, the only matrix element is the norm of the state.
Note that the spectrum of $f f^\dagger$ is real, as each block $f f^\dagger|_k$ is a Hermitian matrix.

Unitary rotations do not change the particle count. Consequently, the block structure is preserved under rotations and, thus, the $\text{Sp}[f f^\dagger|_k]$ are invariants. If the eigenvalues for two states differ, $\text{Sp}[f_1 f_1^\dagger|_k] \neq \text{Sp}[f_2 f_2^\dagger|_k]$, then the two states {\em can not} be related by a unitary rotation of the modes.
The converse is, in general, not true --- states related by complex conjugation (of $f$), so preserving the spectrum, are not necessarily related by linear optics (see \ref{s:three-qubits} for an example).
It, however, remains an open question whether the converse (up to complex conjugation) is true.

Instead of the eigenvalues, we may compute the characteristic polynomial:
\begin{align}
w_k(\lambda) = \det\left[ f f^\dagger|_{k} - \lambda \mathbbm{I} \right].
\label{characteristic.polynomial}
\end{align}
Since its coefficients are in one-to-one correspondence with the spectrum, the method is equally powerful. Moreover, the coefficients of $w_k(\lambda)$ are polynomials in the coefficients of $f$, which is closer in spirit to formulation of Hilbert's theorem. An alternative, but equivalent, route is to investigate the moments $f f^\dagger|_k$,  $\Tr[(f f^\dagger|_k)^l]$. They are in one to one correspondence with the characteristic polynomial $w_k(\lambda)$ by the virtue of Newton identities \cite{Mead1992}.

For $k=1$, the block is related to the single-particle reduced density matrix, i.e.:
\begin{align}
\rho_1 = f f^\dagger|_{k=1} - n \mathbbm{I}.
\end{align}
For $k>1$ we do not recover the reduced $k$-particle density matrix and, as we will show,
$f f^\dagger|_k$
can provide more entanglement invariants than the spectrum of the reduced density matrices with those respective particle numbers.

Even the first block can give interesting results. We can show that no-go observation for deterministically changing one Fock state into another using with linear optics.

Let us look at $f f^\dagger|_1$. As it is a Fock state, its matrix is diagonal (i.e terms $\langle \psi | a_i a_j^\dagger |\psi \rangle$ vanish for $i\neq j$). The diagonal values, and therefore the eigenvalues, are $\bra{\psi} a_i a_i^\dagger \ket{\psi} = n_i + 1$. As their are invariants, two Fock states can be deterministically related by linear optics if and only if they have the same photon counts (up to a permutation of modes).

\subsection{Correlators and reconstruction}

Knowledge of $f f^\dagger |_k$ for all block particle numbers $k\leq n$ suffices to reconstruct the state $f^\dagger\ket{\Omega}$. The reconstruction strategy is to build the matrix elements of the corresponding density matrix
\begin{align}
\rho_{\vec{n} \vec{n}'} =
\bra{\psi} a_{\vec n}^\dagger a_{\vec n'} \ket{\psi},
\end{align}
which can be done by using the commutation relations in order to express the anti-normally ordered terms into terms with normal ordering.

However, we do not claim that higher blocks with $k>n$ are not important. While they are not required to reconstruct the state, there might be pairs of states whose polynomials $w_0$ up to $w_n$ coincide, yet their $w_k$ differ for some $k>n$.
That is, eigenvalues do not capture relative orientation of eigenvectors for different blocks. Eigenvalues for $k>n$ might incorporate relations between eigenvectors for $k\leq n$.

Let us provide a more straightforward way to reconstruct the state, which does not involve calculating inverting the normal ordering of the operators. Let us recall the notion of {\em frame representation} of a many qudit state \cite{Ferrie2011}. Let $\{\sigma^i\}$ be an orthogonal (in trace norm) set of generators of $SU(d)$ plus the identity (i.e. a basis for $d\times d$ Hermitian matrices). For $SU(2)$ we may just choose the Pauli matrices: $\{\mathbbm{I},\sigma^x,\sigma^y,\sigma^z \}$. Any density matrix of a $n$-qudit state can be written as:
\begin{align}
\rho = \sum_{i_1,\cdots,i_n} t_{i_1 i_2 \ldots i_n}
\sigma^{i_1} \otimes \sigma^{i_2} \otimes \ldots \otimes \sigma^{i_n}
\equiv \sum_{\vec\imath} t_{\vec\imath} \sigma^{\vec\imath},
\label{frame.rep}
\end{align}

Note that for permutation-symmetric states, $t_{i_1 i_2 \ldots i_n}$ must be permutation-symmetric. Since the $\{\sigma^i\}$ are orthogonal, the state can be reconstructed from the expectation values of {\em strings} of $\sigma^{i}$ operators:
\begin{align}
t_{i_1 i_2 \ldots i_n} = \frac{1}{2^n}\Tr\left[ \sigma^{i_1} \otimes \sigma^{i_2} \otimes \ldots \otimes \sigma^{i_n} \;\rho \right].
\label{reconstruction}
\end{align}

Expectation values of permutation-symmetric strings of $\sigma^i$ can be obtained from the correlators $ff^\dagger|_k$, as shown in Appendix (\ref{app:schwinger}). The idea behind the proof is the use of a Schwinger-like representation,
related to the one for spin systems ---see \cite[Chapter 7.2]{Auerbach1994}, and develop identities of the form
\begin{align}
\bra{\psi} \left( \sum_{perm} \sigma^{\vec{\imath}} \right) \ket{\psi} = \bra{\Omega} f A(\vec\imath) f^\dagger \ket{\Omega},
\end{align}
where $A(\vec\imath)$ is a polynomial in creation and annihilation operators. From a practical perspective it allows calculating the expectation value without immersing everything in the full Hilbert space of distinguishable particles,  which has a very high dimension.

For example, for $d=2$, we get the following relation
\begin{gather}
\bra{\psi}
\sum_{perm} (\mathbbm{I})^{\otimes n_{I}}
\otimes (\sigma^x)^{\otimes n_x}
\otimes (\sigma^y)^{\otimes n_y}
\otimes (\sigma^z)^{\otimes n_z}
\ket{\psi}\label{eq:pauli_sym}\\
=\bra{\Omega} f
: \left( a^\dagger a + b^\dagger b \right)^{n_I}
\left( a^\dagger b + b^\dagger a \right)^{n_x}\\
\times \left( - i a^\dagger b + i b^\dagger a \right)^{n_y}
\left( a^\dagger a - b^\dagger b\right)^{n_z} :
f^\dagger \ket{\Omega},
\label{eq:schwinger_form}
\end{gather}
where $n_{I} + n_x + n_y + n_z = n$ (covering all symmetric correlators), the sum is over all $n!$ permutations and :expression: stands for the normal ordering, i.e. putting the creation operators on the left and the annihilation on the right. Note that, for most of this article, we use anti-normal ordering, as we work with operators of the form $f^k f^{\dagger k}$.


\section{Symmetric component of tensor powers}
\label{s:replicas}

An alternative set of invariants can be found by studying the symmetric component of tensor copies of a given multi-photon state, taken in the particle representation.

Typically, $\ket{\psi}_P^{\otimes k}$ is not permutation-symmetric, therefore it does not describe a boson state. However, we will show that its projection on the symmetric subspace is proportional to $f^{\dagger k}\ket{\Omega}$, a $kn$-photon state in $d$ modes.

Let us give an example, with $n=2$ and $d=2$, $\ket{\psi} = \ket{1,1} = {1\over\sqrt{2}} (\ket{12}_P+\ket{21}_P)$. If we multiply it tensorially with itself, we get $\ket{\psi}_P^{\otimes 2} = {1\over2}(\ket{12}_P+\ket{21}_P)\otimes (\ket{12}_P+\ket{21}_P)$. This is {\em not} a valid photon state, because it is {\em not} permutation-symmetric:
\begin{gather}
\tfrac{1}{2} \left( \ket{1212}_P + \ket{1221}_P \right.\label{eq:nonsym-prod}\\ 
\left. + \ket{2112}_P + \ket{2121}_P \right)\nonumber
\end{gather}

Nonetheless, it can be projected on the permutation-symmetric subspace, $\hbox{Sym}^{kn}(\C^d)$.
Let $\P^{(kn)}_{sym}$ stand for that projector, where the upper index represents the number of particles to be symmetrized, in this case --- $kn$.
Then,
\begin{equation}
\bra{\psi}_P^{\otimes 2} \P^{(4)}_{sym} \ket{\psi}_P^{\otimes 2} = {2\over 3}
\end{equation}
because \eqref{eq:nonsym-prod} contains 4 out of 6 possible permutations,
\begin{equation}
\P^{(4)}_{sym} \ket{1212}_P =
\tfrac{1}{6} \left( \ket{1122}_P + \text{permutations} \right).
\end{equation}

In order to make the LU-invariance of those values $\bra{\psi}_P^{\otimes k} \P^{(kn)}_{sym} \ket{\psi}_P^{\otimes k}$ manifest, we will show their relation to
\begin{equation}
\bra{\Omega} f^k f^{\dagger k} \ket{\Omega}
\label{fkfk}
\end{equation}
i.e.: the {\em vacuum expectation values of} $f^k f^{\dagger k}$ for all $k\in\N$. These are easy to compute and their invariance is straightforward, since the vacuum is rotation-invariant. Thus, we will prove the following:
\begin{theorem}%
For every homogeneous polynomial $f$, such that $\ket{\psi} = f^\dagger \ket{\Omega}$, the state generated by its $k$-th power is proportional to the state $\ket{\psi}_P^{\otimes k}$ projected on the fully symmetric space of all particles, that is,
\begin{equation}
f^{\dagger k} \ket{\Omega} = \tfrac{\sqrt{(kn)!}}{\sqrt{(n!)^k}} \P^{(kn)}_{sym} \ket{\psi}^{\otimes k}_P,
\label{eq:fk_and_symmetrization}
\end{equation}
so, in particular:
\begin{equation}
\bra{\Omega}f^kf^{\dagger k} \ket{\Omega} = \tfrac{(kn)!}{(n!)^k} \bra{\psi}_P^{\otimes k} \P^{(kn)}_{sym} \ket{\psi}_P^{\otimes k}.
\label{fkfk.projector}
\end{equation}
\end{theorem}
\begin{proof}
Let $\{\vi^{(1)},\cdots,\vi^{(k)}\}$ be $k$ multi-indices,
denoting photon count at each mode, i.e., for the vector with index $m$, we have 
\begin{equation}
\vi^{(m)} = \{n^{(m)}_1,\cdots,n^{(m)}_d\}.
\end{equation} 
Let us denote by $|\vi^{(m)}|=\sum_l n^{(m)}_l$ the total photon count. The monomial operator defined in \eqref{eq:multia}, $\tilde a^\dagger_{\vi^{(1)}+\cdots+\vi^{(k)}}$, can be written in terms of the individual normalized monomials as
\begin{equation}
\tilde a^\dagger_{\vi^{(1)}} \tilde a^\dagger_{\vi^{(2)}} \cdots \tilde a^\dagger_{\vi^{(k)}} = M(\vi^{(1)},\cdots,\vi^{(k)})\ \tilde a^\dagger_{\vi^{(1)}+\cdots+\vi^{(k)}}
\label{global.creator}
\end{equation}
where
\begin{equation}
M(\vi^{(1)},\cdots,\vi^{(k)})\equiv \prod_{l=1}^d
\sqrt{{ (n_{l}^{(1)}+\cdots+n_{l}^{(k)})!  \over (n_{l}^{(1)})! \cdots (n_{l}^{(k)})! }}
\label{combinatorial.factor}
\end{equation}
is the normalization factor. Let us express $f^{\dagger k}\ket{\Omega}$ as a sum of terms of this kind:
\begin{align}
  (f^\dagger)^k \ket{\Omega} &= \sum_{\vi^{(1)},\cdots,\vi^{(k)}} \alpha_{\vi^{(1)}} \cdots \alpha_{\vi^{(k)}}\; \nonumber \\
&\times M(\vi^{(1)},\cdots,\vi^{(k)}) \;\tilde a^\dagger_{\vi^{(1)}+\cdots+\vi^{(k)}} \ket{\Omega}
\label{eq:fk}
\end{align}
so, the coefficient for $\ket{\vec{I}}\equiv \tilde{a}^\dagger_{\vec{I}} \ket{\Omega}$ is
\begin{equation}
\sum_{\vi^{(1)}+\cdots+\vi^{(k)}=\vec I} \alpha_{\vi^{(1)}}\cdots \alpha_{\vi^{(k)}} \cdot M(\vi^{(1)},\cdots,\vi^{(k)})
\end{equation}
where $\vec I$ is a multi-index for $nk$ photons in $d$ modes.

Now, let us consider the right hand side of \eqref{eq:fk_and_symmetrization}. The tensor product $\ket{\psi}^{\otimes k}$ can be written as:
\begin{equation}
  \ket{\psi}_P^{\otimes k} =
  \sum_{\vi^{(1)},\cdots,\vi^{(k)}} 
  \alpha_{\vi^{(1)}}\cdots \alpha_{\vi^{(k)}}
  \ket{\vi^{(1)}}_P \otimes \cdots \otimes \ket{\vi^{(k)}}_P,
\label{eq:particle-power-k}
\end{equation}
Notice that the action of several partial projections on symmetric subspaces followed by a global projection on the symmetric subspace is equivalent to just the final global projection. Consequently,
\begin{align} 
&\P_{sym}^{kn} \left( \ket{\vi^{(1)}}_P\otimes\cdots\otimes\ket{\vi^{(k)}}_P
  \right) \nonumber \\
&= N(\vi^{(1)})\cdots N(\vi^{(k)}) \P^{(kn)}_{sym} \nonumber\\
&\left(
\P^{(n)}_{sym}(\ket{\vi^{(1)}}_A) \otimes \cdots \otimes
\P^{(n)}_{sym}(\ket{\vi^{(k)}}_A) \right) \\
&= N(\vi^{(1)})\cdots N(\vi^{(k)}) \P^{(kn)}_{sym} \left(
\ket{\vi^{(1)}+\cdots+\vi^{(k)}}_A \right) \nonumber \\
&= {N(\vi^{(1)})\cdots N(\vi^{(k)})\over N(\vi^{(1)}+\cdots+\vi^{(k)})}
\ket{\vi^{(1)}+\cdots+\vi^{(k)}} = \nonumber \\
&= \tfrac{\sqrt{(kn)!}}{\sqrt{(n!)^k}} M(\vi^{(1)},\cdots,\vi^{(k)}) \ket{\vi^{(1)}+\cdots+\vi^{(k)}}
\end{align}

Applying the above relations to \eqref{eq:particle-power-k} we get
\begin{align}
\P_{sym}^{kn} \ket{\psi}^{\otimes k}_P &=
\sum_{\vi^{(1)},\cdots,\vi^{(k)}}
\alpha_{\vi^{(1)}}\cdots \alpha_{\vi^{(k)}}
\tfrac{\sqrt{(kn)!}}{\sqrt{(n!)^k}}\\
&\times M(\vi^{(1)},\cdots,\vi^{(k)})
\ket{\vi^{(1)}+\cdots+\vi^{(k)}}\nonumber
\end{align}
which is a state proportional to \eqref{eq:fk}, with the proportionality factor $\sqrt{(kn)!/(n!)^k}$, thus we have shown \eqref{eq:fk_and_symmetrization}. 

\end{proof}


This tensor product symmetrization trick bears resemblance to the use of Clebsch-Gordan coefficients. Indeed, already for $k=2$ the result is useful: $\ket{\psi}^{\otimes 2}$ is not permutation-symmetric unless $\ket{\psi}=\ket{\phi}^{\otimes n}$ for some single-particle state $\ket{\phi}$.

It is possible to prepare an experimental setup to measure $\langle f^k f^{\dagger k} \rangle$. We have to prepare $k$ copies of the state and project each $k$-tuple of modes into their symmetric combination. For example, if $k=2$, two modes are symmetrized using a beam-splitter. Then, $\langle f^2 f^{\dagger 2}\rangle$ is the probability amplitude for losing no photons in the procedure. In general, taking copies of bosonic states and calculating projections offers a way to measure multi-particle entanglement, since taking $k$ copies provides a way to measure R\'enyi entropy of order $k$ of the given subsystems \cite{Daley2012}.

There is another interpretation of $\langle f^k f^{\dagger k} \rangle$ in polynomial language. The quantity we are investigating is known as the {\em Bombieri norm} of homogeneous polynomials \cite{Beauzamy1990} (in this case, $f^k$), which is known to be invariant under unitary rotations of the variables.
This quantity can be expressed as an integral of $|f(\vec{a})|^{2k}$ over the (complex) unit sphere $|\vec{a}|=1$, \cite{Pinasco2012, Pinasco2005} (equivalently, see \cite[Lemma 15]{Aaronson2010}, where it is called Fock Inner Product).


\section{Tensor product in mode basis and an experimental recipe for $f^{\dagger k}$}
\label{s:mode-product}

In this section we study tensor product in the mode representation $\ket{\psi}_M^{k}$, which is different and more physically relevant than tensor product in the particle representation discussed in the previous section.

Furthermore, we provide experimentally-feasible way do directly measure the invariants $\bra{\Omega} f^k f^{\dagger k} \ket{\Omega}$, defined as in \eqref{fkfk}, as related to success-rate of creation of states $f^{\dagger k}\ket{\Omega}$ from $k$ copies of state $f^\dagger \ket{\Omega}$.

To start with, let us look at example of $n=3$ particles in $d=2$ modes, raised to power $k=2$
\begin{align}
&\left( \tfrac{1}{\sqrt{2}} (\ket{0,3}_M+\ket{2,1}_M) \right)^{\otimes 2}\\
&= \tfrac{1}{2}\left( \ket{0,3,0,3}_M + \ket{0,3,2,1}_M\right.\\
&\left.+ \ket{2,1,0,3}_M+\ket{2,1,2,1}_M \right).
\end{align}
This is a valid photon state (as permutation-symmetry of particles is built-in in the mode representation), of $6$ particles in $4$ modes. 

In general, raising a bosonic state to tensor power, in the mode representation, yields in $kn$ photons in $kd$ modes (not $kd$ particles in $d$ modes, as in the tensor power for particle representation).

Tensor product in mode representation has a direct physical interpretation. If we create $k$ optical tables the same setups, each one producing state $\ket{\psi}$, then $\ket{\psi}_M^{\otimes k}$ is the quantum state produced by the laboratory.
As we see, multiplying state also multiplies number of modes, as there is one more parameter related to the number of optical table.

The question is if it this product can be related to $f^{\dagger k}$ in some way?  The answer is positive.  This time instead of symmetrizing particles (as we did for $\ket{\psi}_P^{\otimes k}$) we need to reduce number of modes from $kd$ to $d$,
by performing some symmetrization of modes.

We can write
\begin{align}
\ket{\psi}_M^{\otimes k} &= f^\dagger(a_{(1,1)},\ldots,a_{(d,1)})\label{eq:fkmanymodes}\\
&\times f^\dagger(a_{(1,2)},\ldots,a_{(d,2)}) \times \ldots\\
&\times f^\dagger(a_{(1,k)},\ldots,a_{(d,k)})\ket{\Omega}.
\end{align}
That is, if we are taking a number of copies of a bosonic state, then we in fact multiply number of modes.
The second index is related to copy.

System is symmetrized with respect to particles inside mode, by construction.
To symmetrize among modes, we need to project it on symmetric combination of respective modes
\begin{align}
b_{(i,1)} = \frac{a_{(i, 1)} + \ldots + a_{(i, k)} }{\sqrt{k}},\label{eq:sym_modes}
\end{align}
where all $b_{(i,j)}$ need to be pairwise orthogonal. 
It can be realized with linear optics, as unitary rotation of modes.
In particular, we may employ Fourier transform (i.e. $\vec{b}_i = \mathcal{F} \vec{a}_i$ for each group of modes), and we are interested in the constant term.

When inverting Fourier transform, each mode can be expressed as a linear combination of $b_{(i,j)}$,
where states with different indices are orthogonal,
and weight of $b_{(i,1)}$ is always $1/\sqrt{k}$.

Consequently,
\begin{align}
&f^\dagger(a_{(1,j)}, \ldots, a_{(d,j)})\label{eq:symmetrization_of_variables}\\
= &f^\dagger \left(\tfrac{1}{\sqrt{k}} b_{(1,1)} + \text{orth.}, \ldots, \tfrac{1}{\sqrt{k}} b_{(d,1)} + \text{orth.} \right)\\
= & k^{-n/2} f^\dagger(b_{(1,1)}, \ldots, b_{(d,1)}) + \text{orth.},
\end{align}
where by \emph{orth.} we denote terms containing at least one $b_{(i,j\neq 1)}$.

Thus, by using \eqref{eq:symmetrization_of_variables} for every component of \eqref{eq:fkmanymodes} we get
\begin{align}
k^{-kn/2} f^{\dagger k}(b_{(1,1)}, \ldots, b_{(d,1)}) + \text{orth}.
\end{align}
\begin{figure}[!htbp]
	\centering
		\includegraphics[width=0.3\textwidth]{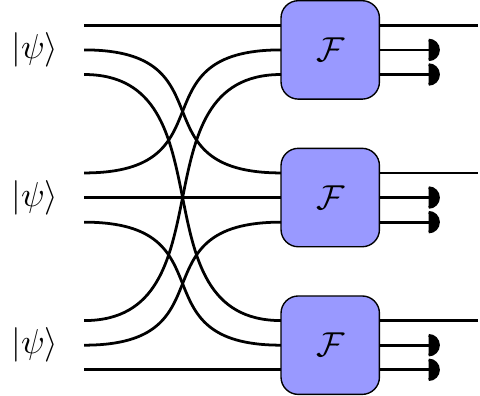}
	\caption{Experimental setup example for $d=3$ modes (and $\mathcal{F}$ operators) and $k=3$ copies (and outcome channels per operator).}
	\label{fig:fpower}
\end{figure}

Consequently, we have one more interpretation of $f^{\dagger k}\ket{\Omega}$.
It is the state you get when following the recipe, pictured in Fig.~\ref{fig:fpower}:
\begin{itemize}
\item Create $k$ copies of an $n$-photon state.
\item Perform interference on each group of respective modes.
\item Postselect results in which for each group of modes no photon was detected in non-first output mode.
\end{itemize}

Our probability to succeed is
\begin{equation}
\frac{\bra{\Omega} f^k f^{\dagger k} \ket{\Omega}}{k^{kn}}
\leq \frac{(kn)!}{(n!)^k k^{kn}}
\approx k^{-1/2} (2 \pi n)^{(1-k)/2},
\end{equation}
where the approximation is due to Stirling's approximation for $(kn)!$ and $n!$.
That is, invariant $\bra{\Omega} f^k f^{\dagger k} \ket{\Omega}$ can be measured experimentally, as statistic of no clicks in detectors, in the described setting.

For the simplest case of $n=1$, $d=1$ and $k=2$, the Fourier transform becomes%
\begin{align}
\mathcal{F} =
\left[%
\begin{matrix}
\tfrac{1}{2} & \tfrac{1}{2}\\
- \tfrac{1}{2} & \tfrac{1}{2}
\end{matrix}
\right].
\end{align}
and we get Hong-Ou-Mandel interference with post-selection, allowing us to produce state two photons in one mode $\ket{2,0}_M$ from two photons in two modes $\ket{1,1}_M$, with $50\%$ post-selection efficiency.

Moreover, a similar experimental scheme as above can be used to produce states of the from%
\begin{align}
f_1^\dagger \cdots f_k^\dagger \ket{\Omega},
\end{align}
where all $f_i^\dagger \ket{\Omega}$ are states of a fixed number of photons (perhaps different for each $i$).
The success rate is%
\begin{align}
\frac{\bra{\Omega} f_k \cdots f_1  f_1^\dagger \cdots f_k^\dagger  \ket{\Omega}}{k^{n_1 + \cdots + n_k}}.
\end{align}
This follows directly from \eqref{eq:symmetrization_of_variables} applied to a product of functions.


\section{Examples at work}
\label{s:examples}

The previous two sections have introduced two sets of LU-invariants for $n$-photon states in $d$-modes. The question to be addressed in this section is the following: can those invariants help us determine the LU-equivalence classes of relevant states? We will start our discussion with a benchmark problem, which can be solved in many different ways: $n=2$ photons in $d=2$ modes. Then, we will proceed to the case of $n=3$ particles, still in $d=2$ modes, which is the first non-trivial case, although it is well understood. We will show that, in that case, the right number of polynomial invariants is recovered. Our last example is a much more complicated system: $n=4$ photons in $d=8$ modes with some additional symmetries.

\subsection{2 particles in 2 modes}

The simplest example is $n=2$ particles in $d=2$ modes:
\begin{align}
f = \alpha_{20} \tfrac{a_1^{2}}{\sqrt{2}} + \alpha_{11} a_1 a_2 + \alpha_{02} \tfrac{a_2^{2}}{\sqrt{2}}.
\end{align}
There is just a single invariant. Let us study how we can obtain it using the methods described in this paper. In our case it suffices to look at a block of $k=1$ particles:
\begin{align}
\left[%
\begin{matrix}
3|\alpha_{20}|^2 + 2 |\alpha_{11}|^2 + |\alpha_{02}|^2 & 
\sqrt{2}(\alpha_{20}^\star \alpha_{11} + \alpha_{11}^\star \alpha_{02}) \\
\sqrt{2}(\alpha_{20} \alpha_{11}^\star + \alpha_{11} \alpha_{02}^\star) & 
|\alpha_{20}|^2 + 2 |\alpha_{11}|^2 + 3|\alpha_{02}|^2 
\end{matrix}
\right]
\end{align}
Its characteristic polynomial is
\begin{equation}
w_2(\lambda) = \lambda^2
- \Tr \left( f f^\dagger|_{k=1} \right) \lambda
+ \det \left( f f^\dagger|_{k=1} \right),
\end{equation}
where coefficients are
\begin{align}
\Tr \left( f f^\dagger|_{k=1} \right)
&= 4\left(|\alpha_{20}|^2 + |\alpha_{11}|^2 + |\alpha_{02}|^2\right),
\nonumber \\ 
\det \left( f f^\dagger|_{k=1} \right)
&= 4 \left(|\alpha_{20}|^2 + |\alpha_{11}|^2 +
|\alpha_{02}|^2\right)^2 \nonumber \\ 
- ( |\alpha_{20}|^2 &- |\alpha_{02}|^2 )^2 + 2 |\alpha_{20}^\star \alpha_{11} +
\alpha_{11}^\star \alpha_{02}|^2. 
\label{eq:twotwodet} 
\end{align}
The trace gives only the normalization, which is the same information contained in $f f^\dagger|_{k=0}$, and which we can set to $1$. The determinant, on the other hand, gives a new invariant.

Alternatively, we can factorize the (degree 2) polynomial: $f=f_1f_2$. In other terms, we can make use of the Majorana stellar representation:
\begin{align}
\ket{\psi} &= \frac{1}{\sqrt{N}}f_1^\dagger f_2^\dagger
\ket{\Omega}\\ & = \frac{1}{\sqrt{2N}}\left(
\ket{\phi_1}_P \otimes \ket{\phi_2}_P
+ \ket{\phi_2}_P \otimes \ket{\phi_1}_P \right),
\end{align}
where $\sqrt{N}$ is a normalization factor and $f_i^\dagger \ket{\Omega} \equiv \ket{\phi_i}$. Since $U\in SU(2)$ acts on the representation as a simultaneous rotation of the points, for two particles the only invariant is the angle between the states, or equivalently $|\braket{\phi_1}{\phi_2}|^2$. A straightforward (albeit tedious) calculation gives
\begin{align}
|\braket{\phi_1}{\phi_2}|^2 = 
\frac{|\alpha_{20}|^2 + |\alpha_{11}|^2 + |\alpha_{02}|^2 - |\alpha_{11}^2 - 2 \alpha_{20} \alpha_{02}|}
{|\alpha_{20}|^2 + |\alpha_{11}|^2 + |\alpha_{02}|^2 + |\alpha_{11}^2 - 2 \alpha_{20} \alpha_{02}|}.
\end{align}
Along with the normalization condition it yields the invariant
\begin{align}
|\alpha_{11}^2 - 2 \alpha_{20} \alpha_{02}|^2 
&= 3  - \det \left( f f^\dagger|_{k=1} \right).
\end{align}
The above is $0$ and $1$ for orthogonal and parallel vectors $\ket{\phi_i}$, respectively.

It is also possible to find the $\langle f^k f^{\dagger k}\rangle$ invariants associated to $k$ copies. For $k=2$ we obtain:
\begin{align}
\tfrac{2^2}{4!}\bra{\Omega} f^2 f^{\dagger 2} \ket{\Omega} = 1 -
\tfrac{1}{3}|\alpha_{11}^2 - 2 \alpha_{20} \alpha_{02}|^2.
\end{align}
In particular, for each orbit under linear optics, we can give a representative, for example
\begin{align}
\frac{\cos(\theta)}{\sqrt{2}} a_1^2 + \frac{\sin(\theta)}{\sqrt{2}} a_2^2 
\end{align}
for $\theta\in[0,\frac{\pi}{4})$.

\subsection{Three qubits}
\label{s:three-qubits}

The case of $n=3$ photons in $d=2$ modes can be viewed as three qubits in a permutation-symmetric state, and is more involved. A full list of invariants is listed in \cite{Sudbery2000}. Disregarding mirror-reflection (i.e.: anti-unitary operators) there are 6 invariants, which reduce to 4 when we take into account normalization and permutation-symmetry. A normal form can be employed \cite{acin2000generalized, Acin2001a, Carteret2000} which, when particularized to a permutation-symmetric state, gives
\begin{align}
\ket{\psi} = &p \left(\ket{001}_P+\ket{010}_P+\ket{100}_P\right)/\sqrt{3} \nonumber\\ +
  &q \ket{111}_P + r \exp(i \varphi) \ket{000}_P,
\label{eq:acin3form}
\end{align}
where all parameters ($p$, $q$, $r$, $\varphi$) are real.
In this section, we use modes $\{0,1\}$, which are more prevalent in description of qubits, $\{1,2\}$ (in most of this paper we start enumeration from $1$).
Or, in polynomial notation:
\begin{align}
f = \alpha_{30} \tfrac{a_0^3}{\sqrt{6}}
  + \alpha_{21} \tfrac{a_0^2 a_1}{\sqrt{2}}
  + \alpha_{03} \tfrac{a_1^3}{\sqrt{6}},
\end{align}
where $\alpha_{30}$ is complex and both $\alpha_{21}$ and $\alpha_{03}$ are real parameters.

Our main result is that both the set of moments $\langle f^k f^{\dagger k}\rangle$ with $k\leq 5$ and the characteristic polynomials of the blocks $ff^\dagger|_{k\leq 2}$ {\em provide all invariants}. This can be checked by computing the matrix of partial derivatives of these invariants with respect to the parameters determining state \eqref{eq:acin3form} at, e.g., the point $(p=q=r=1, \varphi=\pi/4)$, and observing that is has maximal rank.

This result implies that blocks of $ff^\dagger$ convey more information than reduced density matrices, which are known to provide only 2 invariants, including the normalization
(note that for $1$ qubits spectra of one-particle and two-particle reduced density matrix are the same).

Beyond this dimensionality test, it is relevant to test whether those invariants can distinguish between states related by complex conjugation (or reflection, in terms of the Majorana representation), i.e.: $\ket{\psi}$ and $\ket{\psi}^*$. In general, for $n \geq 3$, such states do not need to be related by a unitary transformation
(as, in the Majorana representation, 3 indistinguishable unit vectors need tot to have mirror symmetry).
Unfortunately, neither moments nor block spectra can distinguish a state from its complex conjugate (as we already noted in Sec. \ref{s:block-decomposition}).

\subsection{Four-particle singlet state}

As a more interesting example we consider $n=4$ photons in $d=8$ modes, composing four qubits whose singlet-subspace determines a logical qubit, see Fig.~\ref{fig:4photons8modes}. There are three Hilbert spaces that are relevant for this scenario: the total Hilbert space $\mathcal{S}^8_4$, the 4-qubit subspace $\mathcal{H}_4$, and the two-dimensional singlet subspace $\mathcal{H}_s$, which determines the logical qubit, structured by the following inclusions:
\begin{equation}
\mathcal{S}^8_4 \supset \mathcal{H}_4 \supset \mathcal{H}_s \ .
\label{eq:singlet-subsubspace}
\end{equation}

We address here the following natural question: starting with a particular singlet state $\ket{\psi}\in\mathcal{H}_s$, which singlet states (also in $\mathcal{H}_s$) can be obtained from it using only linear optics? Before proceeding further, let us first describe the details of the construction of the 4-qubit and the singlet subspaces of $\mathcal{S}^8_4$.  

Let us denote the by $\{a_i,b_i\}_{i=1}^4$ the four annihilation operators required to span $\mathcal{S}^8_4$, where the $a_i$ refer to horizontal and the $b_i$ to vertical polarizations of the $i$-th beam. We define the 4-qubit subspace $\mathcal{H}_{4}\subset \mathcal{S}^{8}_4$, as a subspace spanned by states that have exactly one particle in each of the four pairs of modes: $(a_i,b_i)$. This subspace has dimension 16 and is isomorphic to the Hilbert space of four distinguishable qubits $\left(\C\right)^{\otimes 4}$. Action of the local unitary group $SU(2)^{\otimes 4}$ on $\mathcal{H}_4$ is modeled by the action of global linear optics operations that do not mix pairs $(a_i,b_i)$. The two-dimensional singlet subspace, $\mathcal{H}_s$, is defined as the subspace of $\mathcal{H}_4$, which is invariant under the action of any collective unitary rotations on all four qubits, $V^{\otimes 4}$. 

The above construction was first introduced in \cite{Zanardi1997} as the simplest example of a decoherence-free subspace for collective rotations, and it has been created experimentally \cite{Weinfurter2001}. In \cite{Migdal2011dfs} it was shown that the logical qubit is immune to one-particle loss and a protocol for quantum key distribution using such states and linear optics was provided.
\begin{figure}[!htbp]
\centering
\includegraphics[width=0.35\textwidth]{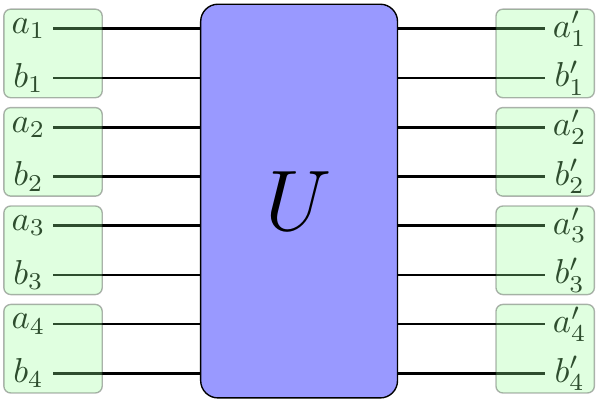}
\caption{Linear transformations for a state with 4 photons distributed among 8 modes, $\mathcal{S}^8_4$. We consider states having exactly one photon in each pair of nodes (denoted by green boxes).
This subspace is equivalent to the Hilbert space of 4 distinguishable particles, $\mathcal{H}_4$.
Furthermore, we study singlet states, i.e. states that are invariant with respect to $U=V^{\otimes 4}$, for all unitary $V$, where each $V$ acts on the respective pair of modes.}
\label{fig:4photons8modes}
\end{figure}

Let us describe the structure of the singlet space in the mode description. For each pair of beams we can define the two-photon {\em singlet} state:
\begin{align}
s_{12} = \left( a_1 b_2 - b_1 a_2 \right) / \sqrt{2},
\end{align}
i.e. $s^\dagger_{12}\ket{\Omega} = \left( \ket{HV} - \ket{VH} \right)/\sqrt{2}$, where $\ket{H}$ and $\ket{V}$ stand for horizontal and vertical polarization, respectively. Those two-photon singlet states can be paired in three inequivalent ways in order to build a global $n=4$ state:
\begin{align}
s_{12}s_{34}, \quad s_{13}s_{42}, \quad s_{14}s_{23}.
\label{eq:pairs_of_pairs}
\end{align}
These three states are not orthogonal, since they span a two-dimensional subspace. In fact, the ordering of particles in $s_{13}s_{42}$ was selected so that the scalar product between each pair is $-1/2$. To form an orthogonal basis, we prepare two linear combinations of them, resembling circular polarization states:
\begin{align}
l &= \tfrac{\sqrt{2}}{3}(s_{12}s_{34} + \epsilon s_{13}s_{42} + \epsilon^2 s_{14}s_{23})\label{eq:singlet_lr_def}\\
r &= \tfrac{\sqrt{2}}{3}(s_{12}s_{34} + \epsilon^2 s_{13}s_{42} + \epsilon s_{14}s_{23}),
\end{align}
where $\epsilon = \exp(i 2\pi/3)$.

Let us introduce the following parametrization for our state
\begin{align}
f = \cos(\tfrac{\theta}{2}) l + \sin(\tfrac{\theta}{2}) e^{i\varphi} r,
\label{eq:singlet_lr_basis}
\end{align}
where $\theta\in[0,\pi)$ and $\varphi\in[0,2\pi)$, so that we can absorb the sign in $\theta$. As it is a logical qubit (i.e. a two dimensional Hilbert space), it can be represented on the Bloch sphere, see Fig.~\ref{fig:2_dims_singlet_geometry}.
\begin{figure}[!htbp]
\centering
\includegraphics[width=0.40\textwidth]{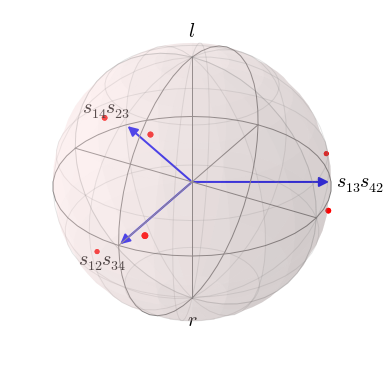}
\caption{Arrows stand for $s_{12}s_{34}$, $s_{13}s_{42}$ and $s_{14}s_{23}$. On the poles there are $l$ and $r$ states, as defined in \eqref{eq:singlet_lr_def}. Points represent a single state subjected to action related to all permutations of pairs of modes modes.}
\label{fig:2_dims_singlet_geometry}
\end{figure}

Now let us compute the moments up to a few copies:
\begin{align}
\langle f^{2} f^{\dagger 2} \rangle &= \tfrac{17}{2} - \tfrac{1}{2}
\cos(2\theta), \nonumber \\
\langle f^{3} f^{\dagger 3} \rangle &= 290 - 42 \cos(2\theta) - 8
\sin^3(\theta) \cos(3 \varphi), 
\end{align}
as a side note, the normalization factors (as in \eqref{fkfk.projector}) are $1/70$ and $1/34650$, respectively. I.e.: the states are very far from being coherent. 

Consequently, we obtain two invariants:
\begin{align}
\cos(2\theta) \quad \text{and} \quad \cos(3 \varphi).
\label{eq:singlet_invariants}
\end{align}
This results restricts the allowed operations within linear optics.
If we restrict ourselves further, only to operations preserving the singlet subspace, then the only possible operations,
in the Bloch representation
(see Fig.~\ref{fig:2_dims_singlet_geometry}) are: rotation along the equator by $2\pi/3$ and $4\pi/3$, rotation around states \eqref{eq:pairs_of_pairs} by $\pi$, and mirror reflection with respect the equatorial plane. In particular, there are no continuous allowed transformations \cite{Wasilewski2007} for such singlet states.
Let show how to implement all those operations, with the exception of the mirror reflection.

What are the possible operations which hold the state within the singlet subspace? Of course, different parings can be interchanged by permuting beams. For example, $(2 \leftrightarrow 3)$ changes $s_{12}s_{34}$ into $-s_{13}s_{24}$ (and the same changing $(1 \leftrightarrow 4)$). Exchange of any two particles acting on any of the three two-singlet parings produces a state with a minus sign. Thus, permuting particles preserves the singlet subspace. 

The group of permutations of $4$ particles has $24$ elements, which can be generated by two-particle swaps:
\begin{align}
(1 \leftrightarrow 2) \text{ or } (3 \leftrightarrow 4):\quad & 
&l &\mapsto - r,\quad
&r &\mapsto - l\\
(1 \leftrightarrow 3) \text{ or } (2 \leftrightarrow 4):\quad & 
&l &\mapsto - \epsilon^2  r,\quad
&r &\mapsto - \epsilon  l\\
(1 \leftrightarrow 4) \text{ or } (2 \leftrightarrow 3): \quad & 
&l &\mapsto - \epsilon r,\quad
&r &\mapsto - \epsilon^2 l,
\end{align}
which can be checked directly by permuting particles in \eqref{eq:singlet_lr_def}. On the Bloch sphere, they are just rotations by $\pi$ around one of the states \eqref{eq:pairs_of_pairs}.

Composition of two permutations allows us to reach cyclic permutations of the three particles, e.g. ($1 \rightarrow 2 \rightarrow 3 \rightarrow 1$). It turns out that such permutations result in $\phi \mapsto \phi + 2\pi/3$ and $\phi \mapsto \phi + 4\pi/3$.

Thus we reached all operations unitary operations allowed by \eqref{eq:singlet_invariants}, with one exception. It does not cover antiunitary operations (reflections on Bloch sphere $\theta \mapsto \pi -\theta$). Thus, it is still possible that there are linear operations not preserving the singlet subspace that map some states into their complex conjugates. Nonetheless, this computation provides the most systematic study of the geometry of the simplest singlet qubit state implemented with photons, to the best of the authors' knowledge.

Alternatively, we can use the spectrum of $f f^\dagger|_k$ for different values of $k$. It suffices to check the two-particles block, i.e. $f f^\dagger |_2$, which is a $36\times 36$ matrix. The highest degree terms of its characteristic polynomial read:
\begin{align}
w_2(\lambda) &= \lambda^{36}\\
&- \lambda^{35} \tfrac{1}{4}\left( 17139 \cos(2 \theta) \right)\nonumber\\
&+ \lambda^{34} \tfrac{1}{72}\left( 9084959 + 1605 \cos(2 \theta) \right.\nonumber\\
&+ \left. 4 \cos(3 \varphi) \sin^3(\theta) \right) - \ldots,\nonumber
\end{align}
which yield the same invariants as the moments.

\section{Conclusion and future work}\label{s:conclusion}

In this paper we analyzed the problem of which states with a fixed number of photons $n$ in $d$ modes can be related using only linear optics. This problem may be mathematically formulated in terms of which homogeneous polynomials of degree $n$ in $d$ complex variables may be related by a unitary transformation between them.

Our proposal has been to study two kinds of invariants. Both are based on the global creation operator, which creates the state, $\ket{\psi}=f^\dagger(\vec a)\ket{\Omega}$, which can be written as a homogeneous polynomial on the creation operators for each mode. The first set of the invariants is just the spectrum of the operator $f f^\dagger$. The second one is the set of moments of the form $\bra{\Omega}f^k f^{\dagger k} \ket{\Omega}$. This second set of invariants can receive a physical interpretation, since they are related to the probability of not losing particles when $k$ copies of the original state are prepared, and the symmetric channel is post-selected.

The main open question is whether our invariants are fine-grained enough to ensure that if two multiphoton states have the same invariants, they can be connected with linear optics and complex conjugation.
We have computed the invariants for a variety of situations, and found that they provide a complete characterization of the equivalence classes in all of them.
However, this question is not yet answered in the general case. 

Regarding future work, we would like to make the following remarks. First of all, a proof that these invariants provide a full characterization would be very desirable. Or, alternatively, a counterexample, which would lead us to find better invariants.
Second, both methods can be applied for fermions with no modifications beyond changing bosonic by fermionic operators. It deserves investigation whether this method provides new invariants in that case, or whether it simplifies the derivation of already known ones.
A third line of future research will be to extend our results to mixed states, or states without a fixed number of particles. In this last case, moments can still be used,
but the spectral method becomes impractical (as $ff^\dagger$ not longer can be decomposed into blocks). 
But perhaps the most practical open question is: if two multiphoton states $\ket{\phi_1}$ and $\ket{\phi_2}$ can not be related using only linear optics, what is the maximal efficiency for obtaining $\ket{\phi_2}$ out of $\ket{\phi_1}$ using linear optics {\em and} post-selection?

As a final remark, the Bloch sphere plot
(Fig.~\ref{fig:2_dims_singlet_geometry}) was created using the open
source package \emph{QuTiP} \cite{Johansson2012a}, to which one of the
authors (P.M.) has contributed.

\subsection*{Acknowledgments}

P.M. would like to thank David Bar-Moshe \cite{PhysicsSE_BarMoshe2011}, Anton Petrunin \cite{MO_Petrunin2012}, Karol Życzkowski, Antonio Ac\'in and Michał Studziński for fruitful discussions
and acknowledges the Spanish MINCIN/MINECO project TOQATA (FIS2008-00784), EU Integrated Projects AQUTE and SIQS, CHISTERA project DIQUIP, ERC AdG OSYRIS
and the John Templeton Foundation. 
M.O. gratefully acknowledges the support of SFB-TR12 program financed by Deutsche Forschungsgemeinschaft.
J.R.-L. acknowledges grant FIS2012-33642 from the Spanish government.


%

\appendix

\section{Equivalence of states in Majorana representation}
\label{app:majorana}

We show how to decide whether two symmetric $n$-qubit states are equivalent under linear operations.

First, we apply the Majorana's stellar representation to both states, resulting in two sets of vectors, $\{\vec{v}_i\}_{i\in \{1,\ldots,n\}}$ and $\{\vec{u}_i\}_{i\in \{1,\ldots,n\}}$. They may differ by a rotation (i.e. an element of SO(3)) and permutation.

Let us take an ordered pair of two non-parallel vectors $(\vec{u}_1, \vec{u}_2)$.  Then for every ordered pair from the first set $(\vec{v}_i, \vec{v}_j)$ for $i \neq j$, if their scalar products match ($\vec{v}_i\cdot\vec{v}_j=\vec{u}_1\cdot\vec{u}_2$) we explicitly construct a unique rotation that rotates the first pair into the second.

Then we check if such rotation rotates all $\vec{v}_i$ into a distinct vector $\vec{u}_{\sigma{i}}$. If it does, states are equivalent. If for all pairs it does not --- there is not.

As number of ordered pairs of two different vectors is $n^2-n$, the algorithm's complexity is maximal of $n^2$ and complexity of an algorithm for factorization of an $n$-degree polynomial of one variable (to get the Majorana's stellar representation).

\section{Schwinger representation of symmetric operators}\label{app:schwinger}

Below we show that \eqref{eq:pauli_sym} and \eqref{eq:schwinger_form} are the same on permutation-symmetric states.

\subsection{Auxiliary notation}

Let us introduce the following notation:
\begin{align}
a_\mu^\dagger &=
\frac{1}{\sqrt{n+1}}
\sum_{i=0}^{n} \ket{\mu}_i\label{eq:adag_alt_def}\\
a_\mu &=
\frac{1}{\sqrt{n}}
\sum_{i=0}^{n-1} \bra{\mu}_i,
\end{align}
where $\ket{\mu}_i$ means {\it insert $\ket{\mu}$ between $i$-th and $(i+1)$-th particle}, whereas $\bra{\mu}_i$ removes $i$-th particle.
The $n$ is the total number of particles in the state it is acting on.
We show that this notation is consistent, i.e. the left hand sides of \eqref{eq:adag_alt_def} act like creation and annihilation operators, respectively.
However, the right hand side can be applied on any state, not only a permutation symmetric one.

For example:
\begin{align}
&\left( \sum_{i=0}^2 \ket{2}_i \right) \ket{01}_P\\
&= \left( \ket{2}_0 + \ket{2}_1 + \ket{2}_2 \right) \ket{01}_P\\
&= \ket{201}_P + \ket{021}_P + \ket{012}_P
\end{align}
and
\begin{align}
&\left( \sum_{i=0}^2 \bra{2}_i \right) \ket{201}_P\\
&= \left( \bra{2}_0 + \bra{2}_1 + \bra{2}_2 \right) \ket{201}_P\\
&= \braket{2}{2} \ket{01}_P + \braket{2}{0} \ket{21}_P + \braket{1}{2} \ket{20}_P\\
&= \ket{01}_P.
\end{align}

A straightforward check on $n$-particle permutation-symmetric states Dicke state show that this (abuse of) notation makes sense.
That is, let us check that:
\begin{align}
a_\mu^\dagger {\tilde a}^\dagger_{\vec{n}} \ket{\Omega}
&= \left( \frac{1}{\sqrt{n+1}} \sum_{i=0}^n \ket{\mu}_i \right) \ket{\vec{n}},\\
a_\mu {\tilde a}^\dagger_{\vec{n}} \ket{\Omega}
&= \left( \frac{1}{\sqrt{n}} \sum_{i=0}^{n-1} \bra{\mu}_i \right) \ket{\vec{n}}.
\end{align}
We proceed by writing a state in particle representation as in \eqref{eq:fock_in_particle_representation}.
For the convenience, without the loss of generality, let us pick $\mu=1$,%
\begin{align}
&\sqrt{n+1} a_1^\dagger \sqrt{\frac{n!}{n_1! \cdots n_d!}}
\ket{n_1, \cdots, n_d}\\
&= \left( \sum_{i=0}^{n} \ket{1}_i \right)
\left( \ket{1}^{n_1}_P \cdots \ket{d}^{n_d}_P
+ \text{perm.} \right)\\
&= (n_1 + 1)
\left( \ket{1}^{n_1+1}_P \cdots \ket{d}^{n_d}_P
+ \text{perm.} \right)\\
&= (n_1 + 1) \sqrt{\frac{(n+1)!}{(n_1+1)! \cdots n_d!}}
\ket{n_1+1, \cdots, n_d},
\end{align}
where \emph{perm.} means inequivalent permutations.
Factor $(n_1+1)$ in the third line comes from%
\begin{equation}
(n + 1) \frac{n!}{n_1! \cdots n_d!}
\big/ \frac{(n+1)!}{(n_1+1)! \cdots n_d!},
\end{equation}
that is, putting $n+1$ particles and comparing number of inequivalent terms in permutation, for the initial and final state.

And analogously for annihilation:%
\begin{align}
&\sqrt{n} a_1 \sqrt{\frac{n!}{n_1! \cdots n_d!}}
\ket{n_1, \cdots, n_d}\\
&= \left( \sum_{i=0}^{n-1} \bra{1}_i \right)
\left( \ket{1}^{n_1}_P \cdots \ket{d}^{n_d}_P
+ \text{perm.} \right)\\
&=  n
\left( \ket{1}^{n_1-1}_P \cdots \ket{d}^{n_d}_P
+ \text{perm.} \right)\\
&= n \sqrt{\frac{(n-1)!}{(n_1-1)! \cdots n_d!}}
\ket{n_1-1, \cdots, n_d}.
\end{align}
This time $n$ in the third line comes from%
\begin{equation}
n_1 \frac{n!}{n_1! \cdots n_d!}
\big/ \frac{(n-1)!}{(n_1-1)! \cdots n_d!}.
\end{equation}

\subsection{Proof}

We start the proof with the following observation.
When we remove a particle from a symmetric state,
there result does not depend which one
(state of all other particles always permutation symmetric).
That is%
\begin{align}
\bra{\mu}_i \ket{\psi}
= \bra{\mu}_j \ket{\psi}
= \frac{1}{n} \left( \sum_{i=0}^{n-1} \bra{\mu}_i \right) \ket{\psi},\label{eq:annihilation_in_one_place}
\end{align}
where the last equality is a consequence of the former (for an $n$-particle state).

Consequently, when acting on $n$-particle symmetric state we get,
we write subsequent annihilation and creation operators as a single sum:
\begin{align}
&a_{\mu_1}^\dagger \cdots a_{\mu_k}^\dagger
a_{\nu_k} \cdots a_{\nu_1} \ket{\psi}\\
= &\frac{(n-k)!}{n!} 
\sum_{i_1,\ldots,i_k} \sum_{j_1,\ldots,j_k}\\
&\Big(
\ket{\mu_1}_{i_1} \cdots \ket{\mu_k}_{i_k}
\bra{\nu_k}_{j_k} \cdots \bra{\nu_1}_{j_1}
\Big) \ket{\psi}\label{eq:aiaj2aiai}\\
= &\left( \sum_{i_1,\ldots,i_k}
\ket{\mu_1}_{i_1} \cdots \ket{\mu_k}_{i_k}
\bra{\nu_k}_{i_k} \cdots \bra{\nu_1}_{i_1}
\right) \ket{\psi}, 
\end{align}
where instead of the sum over $j_1,\ldots,j_k$ we put $j_p=i_p$ using \eqref{eq:annihilation_in_one_place}.

Note that as creation and annihilation operations add and subtract particles (respectively), indices in a product do refer to different set of particles and need to be carried out iteratively.
That is, summation over $j_p$ goes from $j_p=0$ to $n-p$.

We need to show one more thing:%
\begin{align}
&\left( \sum_{i_1,\ldots,i_k}
\ket{\mu_1}_{i_1} \cdots \ket{\mu_k}_{i_k}
\bra{\nu_k}_{i_k} \cdots \bra{\nu_1}_{i_1}
\right) \ket{\psi}
\label{eq:changing2fixed_order}\\
=
&\left( \sum_{\text{p.d. }l_1,\ldots,l_k}
\ket{\mu_1}_{l_1}\bra{\nu_1}_{l_1}
\cdots \ket{\mu_k}_{l_k}\bra{\nu_k}_{l_k}
\right) \ket{\psi},
\end{align}
where by \emph{p.d.} we mean pairwise different.
In fact the only thing we need to do is to relabel each component of the sum. 
In the first line $i_p\in{0,\ldots, n-p}$, while in the second --- $l_p\in{0,\ldots, n-1}$ but disallow repetitions.
If in the first line we relabel in such a way that we don't forget about particles that we removed with $\bra{\nu_1}_{i_p}$, then we get $l_p$.

When we combine \eqref{eq:aiaj2aiai} with \eqref{eq:changing2fixed_order} we get an important relation%
\begin{align}
&a_{\mu_1}^\dagger \cdots a_{\mu_k}^\dagger
a_{\nu_k} \cdots a_{\nu_1} \ket{\psi}
\label{eq:normal_order_and_symmetric_operators}\\
=&\left( \sum_{\text{p.d. }l_1,\ldots,l_k}
\ket{\mu_1}_{l_1}\bra{\nu_1}_{l_1}
\cdots \ket{\mu_k}_{l_k}\bra{\nu_k}_{l_k}
\right) \ket{\psi}.
\end{align}

After showing relation \eqref{eq:normal_order_and_symmetric_operators}, we proceed to the main part of the proof.
Any symmetrized product of matrices is multilinear in their matrix entries, defined by $((\mu_1, \nu_1),\ldots,(\mu_n,\nu_n))$, where each $\mu_i$ (and $\nu_i$) is in $\{0,\ldots, d-1\}$, that is%
\begin{equation}
\sum_{\vec\imath \in \sigma(\{1,\ldots,n\})}
\ket{\mu_1}_{i_1}\bra{\nu_1}_{i_1}\ldots \ket{\mu_n}_{i_n}\bra{\nu_n}_{i_n}.
\end{equation}
So we need to show that for a sum of distinct matrix elements give the corresponding normally ordered operators.
When we apply \eqref{eq:normal_order_and_symmetric_operators}, we get
\begin{equation}
: a_{\mu_1}^\dagger a_{\nu_1} \ldots a_{\mu_n}^\dagger a_{\nu_n} :,
\label{eq:all-normally}
\end{equation}
what completes the proof.

Bear in mind that in \eqref{eq:all-normally} we get $n$ creation and annihilation operators, regardless of the multi-particle operator we want to use.

When we use only a $k$-particle operator, the formula can be simplified, what we show in the examples.

\subsection{Examples}

Below, for the clarity, we will work with qubits and use $a$ and $b$ for the annihilation operators of $\ket{0}$ and $\ket{1}$, respectively.

First, we see that%
\begin{align}
\sum_{i=1}^n \sigma^x_i &= a^\dagger b + b^\dagger a\\
\sum_{i=1}^n \sigma^y_i &= -ia^\dagger  b + i b^\dagger a\\
\sum_{i=1}^n \sigma^z_i &= a^\dagger a - b^\dagger b,
\end{align}
which is the standard Schwinger representation of operators for symmetric states, where we directly applied \eqref{eq:aiaj2aiai}, e.g. for symmetrized $\sigma^y$%
\begin{align}
\sum_{j=1}^n \sigma^y_j &= 
\sum_{j=1}^n \left( -i \ket{0}_j\bra{1}_j + i \ket{1}_j\bra{0}_j \right)\\
&= -ia^\dagger  b + i b^\dagger a.
\end{align}

Now, let us look at symmetrized product of two operators, e.g. $\sigma^x_i$ and $\sigma^z_j$:%
\begin{align}
&\sum_{i \neq j} \sigma^x_i \otimes \sigma^z_j\label{eq:schwinger_square}\\
&= \sum_{i\neq j}
\left( \ket{0}_i \bra{1}_i + \ket{1}_i \bra{0}_i \right)
\left( \ket{0}_j \bra{0}_j - \ket{1}_j \bra{1}_j \right)\\
&= \sum_{i\neq j}
\left(
\ket{0}_i \bra{1}_i \ket{0}_j \bra{0}_j 
- \ket{0}_i \bra{1}_i \ket{1}_j \bra{1}_j \right.\\
&\phantom{=\sum_{i\neq j}(}\left.
+ \ket{1}_i \bra{0}_i \ket{0}_j \bra{0}_j 
- \ket{1}_i \bra{0}_i \ket{1}_j \bra{1}_j
\right)\\
&= \left(
a^{\dagger 2} a b 
- a^\dagger b^\dagger b^2
+ a^\dagger b^\dagger a^2
- b^{\dagger 2} a b
\right)\\
&= : 
\left( a^\dagger b + b^\dagger a \right)
\left( a^\dagger a - b^\dagger b \right) :
\end{align}
were we applied \eqref{eq:normal_order_and_symmetric_operators} to change summation to creation and annihilation operators.

\end{document}